\documentclass[12pt]{article}
\usepackage[utf8]{inputenc}
\usepackage[T1]{fontenc}

\usepackage{amsmath, amsthm}
\usepackage{amssymb}

\newtheorem{theorem}{Theorem}[section]
\newtheorem{proposition}[theorem]{Proposition}
\newtheorem{corollary}[theorem]{Corollary}
\newtheorem{lemma}[theorem]{Lemma}

\title{Upper bound for the number of closed and privileged words}

\author{Josef Rukavicka\thanks{Department of Mathematics,
Faculty of Nuclear Sciences and Physical Engineering, CZECH TECHNICAL UNIVERSITY
IN PRAGUE
(josef.rukavicka@seznam.cz).}}

\theoremstyle{remark}

\newtheorem{remark}[theorem]{Remark}

\DeclareMathOperator{\Factor}{F}
\DeclareMathOperator{\Alphabet}{A}
\DeclareMathOperator{\bound}{D}
\DeclareMathOperator{\eur}{e}

\date{\small{November 25, 2019}}

\begin{document}
\maketitle

\begin{abstract}
A non-empty word $w$ is a \emph{border} of the word $u$ if $\vert w\vert<\vert u\vert$ and $w$ is both a prefix and a suffix of $u$. A word $u$ with the border $w$ is \emph{closed} if $u$ has exactly two occurrences of $w$. A word $u$ is \emph{privileged} if $\vert u\vert\leq 1$ or if $u$ contains a privileged border $w$ that appears exactly twice in $u$.

Peltomäki (2016) presented the following open problem: ``Give a nontrivial upper bound for $B(n)$'', where $B(n)$ denotes the number of privileged words of length $n$.
Let $\bound(n)$ denote the number of closed words of length $n$. Let $q>1$ be the size of the alphabet. We show that there is a positive real constant $c$ such that \[\bound(n)\leq c\ln{n}\frac{q^{n}}{\sqrt{n}}\mbox{, where }n>1\mbox{.}\]

Privileged words are a subset of closed words, hence we show also an upper bound for the number of privileged words.

\end{abstract}

\section{Introduction}
A non-empty word $w$ is a \emph{border} of the word $u$ if $\vert w\vert<\vert u\vert$ and $w$ is both a prefix and a suffix of $u$. A border $w$ of the word $u$ is the \emph{maximal border} of $u$ if for every border $\bar w$ of $u$ we have that $\vert \bar w\vert\leq \vert w\vert$.  A word $u$ with the border $w$ is \emph{closed} if $u$ has exactly two occurrences of $w$. It follows that $w$ occurs only as a prefix and as a suffix of $u$.  A word $u$ is \emph{privileged} if $\vert u\vert\leq 1$ or if $u$ contains a privileged border $w$ that appears exactly twice in $u$. Obviously privileged words are a subset of closed words.

The properties of closed and privileged words have been studied in recent years  \cite{KeLeSa2013}, \cite{Pelto2013}, \cite{ScSh16}. One of the questions that has been investigated is the enumeration of privileged words.
In \cite{Nicholson2018_priviligedwords}, it was proved that there are constants $c$ and $n_0$ such that for all $n>n_0$, there are at least $\frac{cq^n}{n(\log_q{n})^2}$ privileged words of length $n$. This improves the lower bound for the number of privileged words from \cite{FoJaPeSha2016}. Since every privileged word is a closed word, the result from \cite{Nicholson2018_priviligedwords} forms also a lower bound for the number of closed words.

Concerning an upper bound for the number of privileged words we have found only the following open problem \cite{Pelto2016}: ``Give a nontrivial upper bound for $B(n)$'', where $B(n)$ denotes the number of privileged words of length $n$. 
Also in \cite{Pelto2016}, the author presents an idea how to improve the lower bound from \cite{Nicholson2018_priviligedwords}. On the other hand, in \cite{Pelto2016}, there is no explicit suggestion how to approach the problem of determining the upper bound.

In the current article we construct an upper bound for the number of closed words of length $n$. Since the privileged words are a subset of closed words, we present also a response to the open problem from \cite{Pelto2016}.

We explain in outline our proof. Let $\Alphabet$ be an alphabet with $q>1$ letters, let $\Alphabet^m$ denote the set of all words of length $m$, and let $\Alphabet^*=\bigcup_{m\geq 0}\Alphabet^m$. It is known that $\vert \Alphabet^m\vert=q^m$. Let $\Alphabet_w(n)$ denote the number of words of length $n$ that do not contain the factor $w\in \Alphabet^*$. 
Let $\mu(n,m)$ be the maximal value of $\Alphabet_w(n)$ for all $w$ of length $m$; formally \[\mu(n,m)=\max\{\Alphabet_w(n)\mid w\in \Alphabet^m\}\mbox{.}\] 
Let $\hat \bound(n)$ denote the set of all closed words of length $n$ and let $\hat \bound(n,m)$ denote the set of all closed words of length $n$ having a maximal border of length $m$. Let $\bound(n)=\vert \hat \bound(n)\vert$ and $\bound(n,m)=\vert \hat \bound(n,m)\vert$.

Obviously $\hat\bound(n)=\bigcup_{m=1}^{n-1}\hat\bound(n,m)$ and $\hat\bound(n,m)\cap\hat\bound(n,\bar m)=\emptyset$, where $m\not=\bar m$. 
We show that if $2m>n$ then $\bound(n,m)\leq q^{\lceil\frac{n}{2}\rceil}$ and if $2m\leq n$ then $\bound(n,m) \leq q^m\mu(n-2m,m)$; see Lemma \ref{ukn236e69wl}. It follows that \begin{equation}\label{hhu52s35flo89p}\bound(n)=\sum_{m=1}^{n-1} \bound(n,m)\leq \sum_{m=1}^{\lfloor\frac{n}{2}\rfloor}q^m\mu(n-2m,m)+\sum_{m=\lfloor\frac{n}{2}\rfloor+1}^{n-1}q^{\lceil\frac{n}{2}\rceil}\mbox{.}\end{equation}

Let $\mathbb{N}$ denote the set of positive integers. Let $\omega(n)=\frac{1}{\ln{q}}(\ln{n}-\ln{\ln{n}})$.
Let $\Pi$ denote the set of all functions $\pi(n):\mathbb{N}\rightarrow \mathbb{N}$ such that $\pi(n)\in \Pi$ if and only if $1\leq\pi(n)\leq \max\{1,\omega(n))\}$ and $\pi(n)\leq \pi(n+1)$ for all $n\in \mathbb{N}$. We apply the function $\max$, because $\omega(n)<1$ for some small $n$. 

The key observation in our article is that the number of words of length $n$ that do not contain some ``short" factor of length $\pi(n)\in \Pi$ has the same growth rate as the number of words of length $n-\lfloor\frac{\ln{n}}{\ln{q}}\rfloor$. Formally said, for each $\pi(n)\in \Pi$ there is a positive real constant $c$ such that $\mu(n,\pi(n))\leq cq^{n-\frac{\ln{n}}{\ln{q}}}$; see Theorem \ref{nk211d2qp6h}.
This observation allows us to show that there are real positive constants $c_1, c_2$ such that \begin{equation}\label{ttoek6d6f2v5}\sum_{m=1}^{\lfloor\frac{n}{2}\rfloor}q^{m}\mu(n-2m,m)\leq c_1\ln{n}\sum_{m=\lfloor c_2\ln{n}\rfloor}^{\lfloor\frac{n}{2}\rfloor}q^{m}\mu(n-2m,m)\mbox{.}\end{equation}
In consequence we may count only closed words having a maximal border longer than $c_2\ln{n}$ in order to find an upper bound for $\bound(n)$. Applying that $\mu(n-2m,m)\leq q^{n-2m}$ for $n\geq 2m$, we derive from (\ref{hhu52s35flo89p}) and (\ref{ttoek6d6f2v5}) our result for the number of closed words. 

\section{Upper bound for the number of closed words}

 We present an upper bound for the number of words of length $n$ that avoid some factor of length $m$; it means an upper bound for $\mu(n,m)$.  
\begin{lemma}\label{ujd6562v3b6}
If $n,m\in \mathbb{N}$ then
\[\mu(n,m)\leq q^n\left(1-\frac{1}{q^m}\right)^{\lfloor\frac{n}{m}\rfloor}\mbox{.}\]
\end{lemma}
\begin{proof}
Given $w\in \Alphabet^m$,
let $U_{n,w}$ be a set of words $u=u_1u_2\dots u_{k-1}u_k\in \Alphabet^*$, where $\vert u\vert =n$, $\vert u_i\vert=m$, $w\not=u_i$ for all $1\leq i<k$, and $\vert u_k\vert=n\bmod m\mbox{.}$ It follows that $\vert u_k\vert<m=\vert w\vert$ and thus $u_k\not=w$. Obviously 
\[\vert U_{n,w}\vert=(q^m-1)^{\lfloor\frac{n}{m}\rfloor}q^{n\bmod m}=q^n\left(1-\frac{1}{q^m}\right)^{\lfloor\frac{n}{m}\rfloor}\mbox{.}\] Note that $\vert\Alphabet^m \setminus\{w\}\vert=q^m-1$. It is clear that the set of words of length $n$ not containing the factor $w$ is a subset of $U_{n,w}$. The lemma follows.
\end{proof}

For the proof of Theorem \ref{nk211d2qp6h} we need the following limit.
\begin{proposition}
\label{ir556s8e6w2g6}
We have that \[\lim_{n\rightarrow \infty} n\left(1-\frac{\ln{n}}{n}\right)^n=\eur\mbox{.}\]
\end{proposition}
\begin{proof}
Let
\begin{equation}
\label{uuk21d2e1w2e5}
y=\lim_{n\rightarrow \infty}n\left(1-\frac{\ln{n}}{n}\right)^n \mbox{.}
\end{equation}
From (\ref{uuk21d2e1w2e5}) we have that
\begin{equation}
\label{hki5236d5e6g5}
\ln{y}=\lim_{n\rightarrow \infty}\ln{\left[n\left(1-\frac{\ln{n}}{n}\right)^n\right]}=\lim_{n\rightarrow \infty}\left[ \ln{n}+n\ln{\left(1-\frac{\ln{n}}{n}\right)}\right]\mbox{.}
\end{equation}
Let us consider the second term on the right side of (\ref{hki5236d5e6g5}):
\begin{equation}\label{ggjh233t2y5v}
\begin{split}
\lim_{n\rightarrow \infty}n\ln{\left(1-\frac{\ln{n}}{n}\right)}=\lim_{n\rightarrow \infty}\frac{\ln{\left(1-\frac{\ln{n}}{n}\right)}'}{(\frac{1}{n})'}= \\
\lim_{n\rightarrow \infty}\frac{\frac{(-1)(\frac{1-\ln{n}}{n^2})}{\left(1-\frac{\ln{n}}{n}\right)}}{-\frac{1}{n^2}}=\lim_{n\rightarrow \infty}\frac{n(1-\ln{n})}{n-\ln{n}}\mbox{.}
\end{split}
\end{equation}
Since $\lim_{n\rightarrow\infty}\frac{n}{n-\ln{n}}=1$, it follows from (\ref{hki5236d5e6g5}) and (\ref{ggjh233t2y5v}) that
\[\ln{y}=\lim_{n\rightarrow\infty}\left[\ln{n}+\frac{n(1-\ln{n})}{n-\ln{n}}\right]=\lim_{n\rightarrow\infty}\left[\ln{n}+1-\ln{n}\right]=1\mbox{.}\]
It follows that $y=\eur$. This completes the proof. 
\end{proof}
Let $\mathbb{R^+}$ denote the set of positive real numbers. 

Let $\beta=\frac{1}{\ln{q}}\in\mathbb{R^+}$. The following theorem states that the number of words of length $n$ avoiding some given "short" factor (of length shorter than $\pi(n)\in \Pi$) has the same growth rate as the number of all words of length $n-\beta \ln{n}$. 
\begin{theorem}
\label{nk211d2qp6h}
If $\pi(n)\in \Pi$ then there is a constant $c\in \mathbb{R^+}$ such that for all $n\in \mathbb{N}$ we have that
\[\frac{\mu(n,\pi(n))}{ q^{n-\beta\ln{n}}}\leq c\mbox{.}\]
\end{theorem}
\begin{proof}
From Lemma \ref{ujd6562v3b6} 
we have that
\begin{equation}\label{tthfjug885f4}\frac{\mu(n,\pi(n))}{q^{n-\beta\ln{n}}}=\frac{q^n\left(1-\frac{1}{q^{\pi(n)}}\right)^{\lfloor\frac{n}{\pi(n)}\rfloor}}{q^{n-\beta\ln{n}}}=n\left(1-\frac{1}{q^{\pi(n)}}\right)^{\lfloor\frac{n}{\pi(n)}\rfloor}\mbox{.}
\end{equation}  Realize that $q^{\beta \ln{n}}=n$.

Obviously there is $n_0\in \mathbb{N}$ such that $q^{\pi(n)}\leq \frac{n}{\ln{n}}$ for all $n>n_0$; recall that $\pi(n)\leq \omega(n)= \frac{1}{\ln{q}}(\ln{n}-\ln{\ln{n}})$ as $n$ tends to infinity. Consequently for all $n>n_0$ we have that \begin{equation}\label{bblpe215w96f}n\left(1-\frac{1}{q^{\pi(n)}}\right)^n\leq n\left(1- \frac{\ln{n}}{n}\right)^n\mbox{.}\end{equation}
Proposition \ref{ir556s8e6w2g6} and (\ref{bblpe215w96f}) imply that  \begin{equation}\label{tthfjuf856e6}\lim_{n\rightarrow\infty} n\left(1-\frac{1}{q^{\pi(n)}}\right)^{n}\leq  \eur\mbox{.}\end{equation}
Clearly $\lim_{n\rightarrow \infty}\left(f(n)\right)^{\frac{1}{\pi(n)}}\leq \eur$ for each function $f(n)$ such that $f(n)\geq 0$ and $\lim_{n\rightarrow\infty}f(n)\leq \eur$; recall that $\pi(n)\geq 1$. Then the theorem follows from (\ref{tthfjug885f4}) and (\ref{tthfjuf856e6}). This completes the proof.
\end{proof}
Let $h(n)=\lfloor\beta\ln{n}\rfloor$. 
We present Theorem \ref{nk211d2qp6h} in a slightly different manner that will be more useful for us in the following.
\begin{corollary}
\label{ff21a2w12ppmnw5}
If $\pi(n), \bar \pi(n)\in \Pi$, and $\bar \pi(n)\leq \pi(n)$ then there is a constant $c\in \mathbb{R^+}$ such that for all $n\in \mathbb{N}$ we have that
\[\frac{\mu(n-2\bar \pi(n),\bar \pi(n))}{q^{n-h(n)}}\leq c\mbox{.}\]
\end{corollary}
\begin{proof}
It is easy to verify that $\mu(n-2\bar \pi(n), \bar \pi(n))\leq \mu(n,\pi(n))$, since the number of words of length $n$ avoiding some factor of length $\pi(n)$ is bigger or equal to the number of words of length $n-2\bar \pi(n)$ avoiding some factor of length $\bar \pi(n)\leq \pi(n)$. 

Obviously $h(n)=\lfloor\frac{\ln{n}}{\ln{q}}\rfloor\leq \frac{\ln{n}}{\ln{q}}=\beta \ln{n}$. In consequence we have that $q^{n-h(n)}\geq q^{n-\beta\ln{n}}$. 

The corollary follows from Theorem \ref{nk211d2qp6h}. This completes the proof.
\end{proof}
We show an upper bound for $\bound(n,m)$ for the cases where $2m>n$ and $2m\leq n$.
\begin{lemma} \label{ukn236e69wl} Suppose $n,m\in \mathbb{N}$.
\begin{itemize}
  \item
  If $2m>n$ then $\bound(n,m)\leq q^{\lceil\frac{n}{2}\rceil}$.
  \item
  If $2m\leq n$ then $\bound(n,m)\leq q^m\mu(n-2m,m)$. 
\end{itemize}

\end{lemma}
\begin{proof}
If $2m>n$, $w\in \Alphabet^*$, and $\vert w\vert=m$ then there is obviously at most one word $u$ with $\vert u\vert=n$ having a prefix and a suffix $w$; the prefix $w$ and the suffix $w$ would overlap with each other. If such $u$ exists then the first half of $u$ uniquely determines the second half of $u$. If follows that $\bound(n,m)\leq q^{\lceil\frac{n}{2}\rceil}$.

Let $\Factor(w)$ denote the set of all factors of $w\in \Alphabet^*$. If $n\geq 2m$ then let \[Z(n,m)=\{wuw\mid u\in \Alphabet^{n-2m}\mbox{ and }w\in \Alphabet^m\mbox{ and }w\not\in \Factor(u)\}\mbox{.}\] If $n\geq 2m$ then $\bound(n,m)\subseteq Z(n,m)$. It is easy to see that \[\vert Z(n,m)\vert\leq\vert \Alphabet^m\vert\mu(n-2m,m)\mbox{.}\]
This completes the proof.
\end{proof}

Let $\kappa>1$ be a real constant and $\bar h(n)=\max\{1,\lfloor\frac{1}{\kappa}\omega(n)\rfloor\}$. Again we use the function $\max$ to guarantee that $\bar h(n)\geq 1$ for small $n$. 
\begin{remark}
The function $\bar h(n)$ defines the maximal length of a ``short'' border of a closed word. In the proof of Theorem \ref{hkk2v31c23d} the closed words from $\hat \bound(n,m)$ will be enumerated differently for $m<\bar h(n)$ and for $m\geq\bar h(n)$. 
\end{remark}

The next auxiliary lemma shows an upper bound for $q^{-h(n)+\bar h(n)}$, that we will use in the proof of Proposition \ref{thuut7678yyhty}.
\begin{lemma}
\label{pplk856t59e9e}
There is a constant $c_1\in \mathbb{R^+}$ such that for all $n\in \mathbb{N}$ we have that
\[q^{-h(n)+\bar h(n)}\leq c_1 q^{\frac{1}{\ln{q}}\left(\frac{1}{\kappa}-1\right)\ln{n}}\]
\end{lemma}
\begin{proof}
Let \[y=\lim_{n\rightarrow\infty}(-h(n)+\bar h(n)-\frac{1}{\ln{q}} \left(\frac{1}{\kappa}-1\right)\ln{n})\mbox{.}\]
We have that
\begin{equation}
\begin{split}
y&=\lim_{n\rightarrow\infty}\left(-\lfloor \frac{1}{\ln{q}} \ln{n}\rfloor +\lfloor\frac{1}{\kappa \ln{q}}(\ln{n}-\ln{\ln{n}})\rfloor -\frac{1}{\ln{q}} \left(\frac{1}{\kappa}-1\right)\ln{n}\right)
\\ &=\lim_{n\rightarrow\infty}\left(\frac{\ln{n}}{\ln{q}} \left(-1 + \frac{1}{\kappa}\right)-\frac{1}{\ln{q}} \left(\frac{1}{\kappa}-1\right)\ln{n}\right)
\\ &=0\mbox{.}
\end{split}
\end{equation}
This implies that \[\lim_{n\rightarrow\infty}\frac{q^{-h(n)+\bar h(n)}}{q^{\frac{1}{\ln{q}}\left(\frac{1}{\kappa}-1\right)\ln{n}}}=1\mbox{.}\]
The lemma follows.
\end{proof}
The next proposition shows an upper bound for the number of closed words of length $n$ having a maximal border of length $\leq \lceil\frac{n}{2}\rceil$. 
\begin{proposition}
\label{thuut7678yyhty}There is a constant $c\in\mathbb{R^+}$ such that 
\[\sum_{m=1}^{\lceil\frac{n}{2}\rceil}q^m\mu(n-2m,m) \leq c\ln{n}\frac{q^{n}}{\sqrt{n}}\mbox{, where }n>1\mbox{.}\]
\end{proposition}
\begin{proof}
Since $\mu(n-2m,m)\leq q^{n-2m}$  we have that  
\begin{equation}\label{t12976ukd2f5g}
\sum_{m=1}^{\lceil\frac{n}{2}\rceil}q^m\mu(n-2m,m)\leq \sum_{m=1}^{\bar h(n)-1}q^m\mu(n-2m,m)+\sum_{m=\bar h(n)}^{\lceil\frac{n}{2}\rceil}q^mq^{n-2m} \mbox{.}
\end{equation}
Corollary \ref{ff21a2w12ppmnw5} implies that $\mu(n-2m,m)\leq cq^{n-h(n)}$ for some constant $c\in \mathbb{R^+}$.
It follows that
\begin{equation}\label{ppa586r5op}
\begin{split}\sum_{m=1}^{\bar h(n)-1}q^m\mu(n-2m,m) &\leq \sum_{m=1}^{\bar h(n)}q^m  cq^{n-h(n)}\\ &\leq \bar h(n)q^{\bar h(n)}cq^{n-h(n)}\mbox{.}\end{split}
\end{equation}
Lemma \ref{pplk856t59e9e} and (\ref{ppa586r5op}) imply that 
\begin{equation}
\label{kifu66969t6}
\sum_{m=1}^{\bar h(n)-1}q^m\mu(n-2m,m)\leq c_1\bar h(n)cq^{n-\frac{\ln{n}}{\ln{q}}(1-\frac{1}{\kappa})}\mbox{,}
\end{equation}
where $c_1$ is some real positive constant.

It is easy to verify that \begin{equation}\label{rrjtuk66585g65}\begin{split}q^{-\bar h(n)}\leq q^{-\frac{1}{\kappa \ln{q}}(\ln{n}-\ln{\ln{n}})+1}= q(\ln{n})^{\frac{1}{\kappa}}q^{-\frac{1}{\kappa \ln{q}}\ln{n}} \mbox{.}\end{split}\end{equation} Thus using (\ref{rrjtuk66585g65})
\begin{equation}
\label{rrth2111h25y}
\sum_{m=\bar h(n)}^{\lceil\frac{n}{2}\rceil}q^mq^{n-2m}\leq q^n\sum_{m=\bar h(n)}^{\lceil\frac{n}{2}\rceil}q^{-m}\leq \frac{q^{n-\bar h(n)}}{1-q^{-1}}\leq \frac{q(\ln{n})^{\frac{1}{\kappa}}q^{n-\frac{1}{\kappa \ln{q}}\ln{n}}}{1-q^{-1}}\mbox{.}
\end{equation}
Obviously $\bar h(n)\leq \frac{\ln{n}}{\kappa\ln{q}}$. Hence taking $\kappa=2$, we get  from (\ref{t12976ukd2f5g}), (\ref{kifu66969t6}), and (\ref{rrth2111h25y}) that
\begin{equation}
\label{ttr5t8t4rr66f}
\begin{split}
\sum_{m=1}^{\lceil\frac{n}{2}\rceil}q^m\mu(n-2m,m) &\leq c_1\bar h(n)cq^{n-\frac{1}{2\ln{q}}\ln{n}}+\frac{q(\ln{n})^{\frac{1}{2}}q^{n-\frac{1}{2 \ln{q}}\ln{n}}}{1-q^{-1}}
\\ &\leq q^{n-\frac{1}{2\ln{q}}\ln{n}}\left( c_1c\frac{\ln{n}}{2\ln{q}}+\frac{q(\ln{n})^{\frac{1}{2}}}{1-q^{-1}}\right)
\\&\leq q^{n-\frac{1}{2\ln{q}}\ln{n}}(c_2\ln{n}+c_3(\ln{n})^{\frac{1}{2}})\mbox{,}
\end{split}
\end{equation}
for some constants $c_2,c_3\in \mathbb{R^+}$. Since $\sqrt{n}=q^{\frac{1}{2\ln{q}}\ln{n}}$ the proposition follows from  (\ref{ttr5t8t4rr66f}).
\end{proof}

We show an upper bound for $\bound(n)$.
\begin{theorem} \label{hkk2v31c23d}There is a constant $c\in\mathbb{R^+}$ such that 
\[\bound(n)\leq c\ln{n}\frac{q^{n}}{\sqrt{n}}\mbox{, where }n>1\mbox{.}\]
\end{theorem}
\begin{proof}
We have that 
\begin{equation}
\label{kjk11jj56un5}
\bound(n)=\sum_{m=1}^{n-1}\bound(n,m)= \sum_{m=1}^{\lceil\frac{n}{2}\rceil}\bound(n,m) + \sum_{m=\lceil\frac{n}{2}\rceil+1}^{n-1} \bound(n,m)\mbox{.}
\end{equation}
From Lemma \ref{ukn236e69wl} and (\ref{kjk11jj56un5}) we get that 
\begin{equation}
\label{uu2c5bpw8q948}
\bound(n)\leq \sum_{m=1}^{\lceil\frac{n}{2}\rceil}q^m\mu(n-2m,m) + \sum_{m=\lceil\frac{n}{2}\rceil+1}^{n-1}q^{\lceil\frac{n}{2}\rceil} \mbox{.}
\end{equation}
Realize that \[\sum_{m=\lceil\frac{n}{2}\rceil+1}^{n-1}q^{\lceil\frac{n}{2}\rceil}\leq\frac{n}{2}q^{\lceil\frac{n}{2}\rceil}\] and 
\[\lim_{n\rightarrow\infty}\frac{nq^{\frac{n}{2}}}{\frac{\ln{n}q^{n}}{\sqrt{n}}}=0\mbox{.}\]
Then it follows that from (\ref{uu2c5bpw8q948}), and Proposition \ref{thuut7678yyhty} that there are constants $c_2,c_3\in \mathbb{R^+}$ such that
\[c_2 \sum_{m=1}^{\lceil\frac{n}{2}\rceil}q^m\mu(n-2m,m) \geq  \sum_{m=\lceil\frac{n}{2}\rceil+1}^{n-1}q^{\lceil\frac{n}{2}\rceil}\mbox{ and }\]
\begin{equation}
\label{kkguikd58d6e}
\bound(n) \leq c_3\sum_{m=1}^{\lceil\frac{n}{2}\rceil}q^m\mu(n-2m,m)\mbox{.}
\end{equation}
The theorem follows from (\ref{kkguikd58d6e}), and Proposition \ref{thuut7678yyhty}
\end{proof}
\begin{remark}
Note that the some of the constants $c, c_1,c_2, c_3$, that we used in our results and in particular in Theorem \ref{hkk2v31c23d}, depend on $q$.
\end{remark}

\section*{Acknowledgments}
The author acknowledges support by the Czech Science
Foundation grant GA\v CR 13-03538S and by the Grant Agency of the Czech Technical University in Prague, grant No. SGS14/205/OHK4/3T/14.

\bibliographystyle{siam}
\IfFileExists{biblio.bib}{\bibliography{biblio}}{\bibliography{../!bibliography/biblio}}

\end{document}